\documentclass{article}
\usepackage[margin=0.9in]{geometry}

\usepackage{soul}
\usepackage{url}
\usepackage[
  colorlinks=true,
  linkcolor=black,
  citecolor=blue!60!black,
  urlcolor=blue!60!black
]{hyperref}
\usepackage[utf8]{inputenc}
\usepackage[small]{caption}
\usepackage{graphicx}
\usepackage{amsmath,amssymb,amsfonts}
\usepackage{amsthm}
\usepackage{booktabs}
\usepackage{algorithm}
\usepackage{algorithmic}
\usepackage[switch]{lineno}
\usepackage{xspace}
\usepackage{xcolor}
\usepackage{natbib}

\usepackage{complexity}
\usepackage{tabularx}
\usepackage{makecell}
\usepackage{thmtools}
\usepackage{thm-restate}

\usepackage{natbib}

\newcommand{\shortcite}[1]{\citeyear{#1}}

\theoremstyle{plain}
\newtheorem{theorem}{Theorem}[section]

\newtheorem{corollary}[theorem]{Corollary}

\usepackage{mdframed}
\usepackage{mathtools}

\theoremstyle{definition} % typical definition-style (bold header, normal font)
\newtheorem{definition}[theorem]{Definition}

\newcommand{\ghafull}{{\sc Minimum Envy Graphical House Allocation}\xspace}
\newcommand{\gha}{{\sc ME-GHA}\xspace}
\newcommand{\fpt}{{\sf FPT}\xspace}
\newcommand{\xp}{{\sf XP}\xspace}
\newcommand{\np}{{\sf NP}\xspace}
\newcommand{\LR}[1]{\left\{#1\right\}}

\newcommand{\val}{\alpha}
\newcommand{\envy}{{\sf envy}}
\usepackage{todonotes}
\usepackage{nicefrac}
\usepackage{authblk}
\title{Minimum Envy Graphical House Allocation Beyond Identical Valuations}

\author[1]{Tanmay Inamdar}
\author[1]{Pallavi Jain}
\author[1]{Pranjal Pandey}
\affil[1]{\small Department of Computer Science and Engineering, Indian Institute of Technology Jodhpur, India.}

\date{}

\begin{document}

\maketitle

\begin{abstract}
     House allocation is an extremely well-studied problem in the field of fair allocation, where the goal is to assign $n$ houses to $n$ agents while satisfying certain fairness criterion, e.g., envy-freeness. To model social interactions, the \emph{Graphical} House Allocation framework introduces a social graph $G$, in which each vertex corresponds to an agent, and an edge $(u, v)$ corresponds to the potential of agent $u$ to envy the agent $v$, based on their allocations and valuations. In undirected social graphs, the potential for envy is in both the directions. In the {\sc Minimum Envy Graphical House Allocation} (\gha) problem, given a set of $n$ agents, $n$ houses, a social graph, and agent's valuation functions, the goal is to find an allocation that minimizes the total envy summed up over all the edges of $G$. Recent work, [Hosseini et al., AAMAS 2023, AAMAS 2024] studied \gha in the regime of polynomial-time algorithms, and designed exact and approximation algorithms, for certain graph classes under identical agent valuations.

   We initiate the study of \gha with non-identical valuations, a setting that has so far remained unexplored. We investigate the multivariate (parameterized) complexity of \gha by identifying structural restrictions on the social graph and valuation functions that yield tractability. 
    We also design moderately exponential-time algorithms for several graph classes, and a polynomial-time algorithm for {binary valuations that returns an allocation with envy at most one when the social graph has maximum degree at most 
one}. 
\end{abstract}

\section{Introduction}

Fair division is an important area at the intersection of economics and computer science that deals with allocating resources to agents in a manner that ensures some notion(s) of fairness. House allocation is an important special case of fair division, where the goal is to allocate $m\geq n$ indivisible items (called \emph{houses}) to $n$ agents, with each agent receiving exactly one house. Fairness is commonly measured using criteria such as envy-freeness, proportionality, pareto-optimality, among others. This model dates back to the work of \cite{shapley1974cores}, and has received significant attention in the fair division community \cite{DBLP:journals/jet/Miyagawa01,DBLP:journals/geb/Miyagawa02a,DBLP:conf/isaac/AbrahamCMM05}. 
%In general, the number of available houses $m$ is at least the number of agents $n$, and fairness can be quantified by considering a number of criteria, including envy-freeness, proportionality, pareto-optimality, among others. 

Envy-freeness (and its relaxations) is by far the most well-studied fairness criteria, even in the general fair division setting. While an overwhelming majority of the literature has considered the setting where each pair of agents is capable of envying each other---of course, depending on the allocation of houses to the two agents, and their respective valuations,
in many real-life scenarios, however, agents are part of social networks consisting of their peers, friends, relatives etc. In such cases, an agent will be concerned about the house (item) allocated to themselves, as compared to those assigned to their social connections, while they are indifferent to the items allocated to complete strangers. The classical house allocation model fails to capture this degree of \emph{locality} of fairness. There is a growing body of literature that addresses this issue by introducing the notion of \emph{social graph}. A social graph is an undirected graph $G = (V, E)$, where each vertex $v \in V$ corresponds to each of the $n$ agents, and an edge between two agents $u$ and $v$ represents that the two agents are socially relevant to each other, implying the potential of envy between them. In other words, an agent is only capable to envy (and be envied by) its neighbors in the social graph. This formulation was first introduced by \cite{DBLP:journals/aamas/BeynierCGHLMW19} in the house allocation setting, and by \cite{DBLP:journals/ai/BredereckKN22,DBLP:journals/ai/EibenGHO23} in the general envy-free allocation setting. However, all of these works had only considered the envy-free allocation setting, which is not always possible.

Hosseini et al.~\shortcite{DBLP:conf/atal/HosseiniPSVV23} introduced the more general problem called \ghafull (\gha in short), where the goal is to find an allocation that minimizes the total envy added up across all the edges. In the input to \gha, we are given a social graph $G = (V, E)$, where $V$ is a set of $n$ agents, a set $H$ of $n$ houses, and $n$ valuation functions $\val_v: H \to \mathbb{Q}_{\ge 0}$ for each agent $v \in V$. We represent an instance of \gha as $(G, H, \LR{\val_v}_{v \in V})$. A bijection $\pi: V \to H$ is called an \emph{allocation}. The envy of an allocation $\pi$ is computed using the following expression: {$\sum_{(u,v) \in E}  \envy(u, v, \pi(u),\pi(v))$, where the envy across the edge $(u, v)$ w.r.t.~$\pi$, is given by $\envy(u, v, \pi(u),\pi(v)) = \max\LR{\val_v(\pi(u)) - \val_v(\pi(v)), 0} + \max\LR{\val_u(\pi(v)) - \val_u(\pi(u)), 0}$. The goal of the \gha problem is to find an allocation $\pi$ that minimizes the envy.} 

Although \cite{DBLP:conf/atal/HosseiniPSVV23} introduced this general model, they only focused the case of \emph{identical valuations} (we also call as identical agents), where the valuation functions $\val_v$ are identical for all agents $v \in V$. They noted that the problem in fact generalizes the well-known \textsc{Optimal Linear Arrangement (OLA)} problem, and is therefore computationally challenging -- i.e., \NP-hard, thereby ruling out polynomial-time algorithms unless $\mathsf{P = NP}$. On the positive side, they designed polynomial-time algorithms when the social graph is a path, cycle, or a clique. In a follow-up work, \cite{DBLP:conf/atal/Hosseini0SVV24} designed polynomial-time approximation algorithms for \gha on different classes of connected graphs, such as trees, planar graphs. However, because of the lower bounds inherited from OLA, the approximation ratios obtained were far from desirable.

\subsection{Our results and contributions}\label{section:results} 
This state of art, leaves a major gap in our understanding of the complexity of \gha\ with non-identical valuations, which we address in this work.

The problem is \NP-hard even when the social graph has maximum degree at most one~\cite{DBLP:journals/aamas/BeynierCGHLMW19}. However, under the additional restriction of identical valuations, it becomes polynomial-time solvable~\cite{DBLP:conf/atal/HosseiniPSVV23}. We  design a polynomial-time algorithm that returns an allocation with envy at most one when the valuations are \emph{binary}, i.e., each valuation function maps a house to $0$ or $1$ (which can be thought of as  ``like'' and ``dislike''), and the social graph is of maximum degree one (Theorem~\ref{thm:bivalued-warmup}). Our result also implies that, for such a social graph, for an odd number of agents, there is envy-free allocation, and our algorithm achieves this. We also show that there are instances with optimal envy one when the number of agents is even.  

Since the problem is \NP-hard for identical valuations (i.e., agents), the next natural question is the complexity when houses are identical and agents need not be identical. The problem is trivially polynomial-time solvable as we can assign any house to an agent because an agent values all the houses equally. This tractability result and the following practical scenario motivate us to study the parameter \emph{number of house types} (defined in the next paragraph).

In the eponymous application of the problem, it is much more likely that each house is characterized by a limited number of features, such as the total area, the number of rooms, and the neighborhood. In many real-world settings, especially in urban housing markets, a large fraction of houses belong to apartment complexes where multiple units are essentially identical in terms of layout, size, amenities, and location. E.g., a single apartment building may contain dozens of units with the same floor plan and specifications, and hence these units are naturally perceived as having the same value by the agents. Such features define the \emph{``type''} of a house, and typically the number of house-types~$\ell$ is small, whereas there is varying—but usually high—multiplicity of houses of each type. Note that, in this case, each agent can have a different valuation for each feature, and therefore they can value each house-type differently. Furthermore, different agents' valuations for a house-type can be completely different; however each agent will value all houses of the same type equally.

As we discussed above, \gha is trivially polynomial-time solvable for $\ell=1$; however, for $\ell=2$, it is \NP-hard due to a reduction similar to~\cite[Theorem 3.2]{DBLP:conf/atal/HosseiniPSVV23} (instead of concentrating half of the values around $0$, we set it to $0$. Similarly, remaining half of the values are set to $1$. We present the proof for completeness in Theorem \ref{thm:me-gha-l2-nphard}). Note that in this reduction, valuations are binary and agents are identical. Thus, the problem is \NP-hard even when number of house types is constant, agents are identical, and valuations are binary. %This shows that the problem is ${\sf W}[1]$-hard\footnote{${\sf W}[1]$-hardness and {\sf FPT} are defined in preliminaries.} with respect to the parameter $\ell$. 
Thus, in the hope of tractability, we explore structural properties of the social graph, when $\ell$ is constant. % and ask whether the 

In this paper, we consider three structural properties of the graph which have been also considered in the literature: (i) treewidth ($\mathrm{tw}$)\footnote{Treewidth is a popular measure of sparsity of the graph, which, at a high level, quantifies how ``tree-like'' a graph is. A formal definition can be found in the preliminaries. %Parameterizations by vertex cover number and treewidth were considered in prior work~\cite{DBLP:journals/ai/EibenGHO23,DBLP:journals/ai/BredereckKN22} in general envy-free allocation.
}, (ii) vertex cover number ($vc$), and (iii) clique modulator ($k$), defined as the distance to a complete graph.
Note that, the problem is \NP-hard even for trees~\cite{DBLP:conf/atal/Hosseini0SVV24}, thus, para-\NP-hard with respect to the treewidth ($\mathrm{tw}$) of the graph, since treewidth of a tree is $1$. This parameter has been considered in the several fair allocation papers, including house allocation~\cite{DBLP:journals/ai/EibenGHO23,DBLP:conf/atal/Hosseini0SVV24,DBLP:journals/ai/BredereckKN22}. For a constant number of house types $\ell$, we design an \fpt algorithm parameterised by $\mathrm{tw}$ (Theorem~\ref{thm:tw}). The problem is also known to be {\sf W}$[1]$-hard with respect to the parameter $vc$~\cite{DBLP:journals/aamas/BeynierCGHLMW19}, the vertex cover number {(a parameter larger than treewidth as treewidth is upper bounded by $vc$)} (Theorem~\ref{thm:gha-vc-typed}). {We design an \fpt\ algorithm parameterised by $vc$ when $\ell$ is treated as a constant. This algorithm also implies that the problem is also \fpt\ parameterised by $\ell + vc$ (Theorem~\ref{thm:gha-vc-typed}).  Note that although vertex cover can be bounded in terms of treewidth, the algorithm of Theorem~\ref{thm:tw} does \emph{not} yield an \fpt\ algorithm parameterised by $\ell + vc$. as it has a factor $n^\ell$ is the running time.} The problem is polynomial-time solvable for complete graphs~\cite{DBLP:conf/atal/HosseiniPSVV23}. In the parameterized complexity, distance from a tractable instance is a natural parameter. It remains an open question whether \gha admits \fpt algorithm with respect to the parameter $k$, the distance to a complete graph. We design an \fpt algorithm when parameterised by $k$, when $\ell$ is constant. This algorithm is also \fpt parameterised by $\ell+k$ (Theorem~\ref{thm:gha-clique-modulator}). This result also implies an {\xp} algorithm with respect to $k$. 

In another direction, in the most general case when $\ell = n$, i.e., when each house is unique, we obtain single-exponential algorithm that runs in time $\mathcal{O}^\star(4^{n+ \mathrm{tw} \log n})$ \footnote{$\mathcal{O}^\star$ notation hides factors that are polynomial in $n$, i.e., an algorithm with running time $T \cdot n^{\mathcal{O}(1)}$ can be written as $\mathcal{O}^\star(T)$. Here, $T$ can be a function of $n$ or other parameters of interest.} by substituting $\ell=n$ in the algorithm parameterized by $\ell+\mathrm{tw}$ (Theorem~\ref{thm:tw}). Specifically, for social graphs of bounded treewidth, such as (disjoint unions of) paths, cycles, and forests, we obtain an $\mathcal{O}^\star(4^n)$ algorithm as treewidth is constant for these classes of graphs, and since the treewidth of planar graphs is bounded by $\mathcal{O}(\sqrt{n})$, we obtain an $\mathcal{O}^\star(4^{n+o(n)})$ algorithm in this case. 
More generally, this result extends to all graph classes admitting ``small'' balanced separators (defined in preliminaries). Recall that if a graph class is \(f(n)\)-balanced-separable, then its treewidth is bounded by \(\mathcal{O}(f(n))\)~\cite{cygan2015parameterized}. Consequently, our
treewidth-based algorithm implies that for any \(f(n)\)-balanced-separable class of social graphs, \gha can be solved in time $\mathcal{O}^\star\left(4^n \cdot 2^{\mathcal{O}(f(n)\log n)}\right)$,
as stated in Corollary~\ref{cor:balsep}. In the special case of trees, we design an improved algorithm that runs in time $\mathcal{O}^\star(2^n)$, if the valuations are polynomially bounded integers. The problem is known to be \NP-hard for all these classes of graphs~\cite{DBLP:conf/atal/Hosseini0SVV24}.  

Next, we note that the disconnectedness of the social graph was the main source of hardness in the reductions shown in \cite{DBLP:conf/atal/HosseiniPSVV23}. On the positive side, in case of identical valuations, they extended their polynomial-time algorithms for paths, cycles, and cliques, to the disjoint unions of the respective graph classes. Specifically, if the social graph is a disjoint union of $r$ connected components, all of which are paths, cycles, or equally-sized cliques, they designed an $\mathcal{O}^\star(r!)$ algorithm for identical valuations. We obtain a substantially generalized form of this result in the case of \emph{non-identical valuations}. Consider a social graph $G$ that is a disjoint union of (an arbitrary number of) connected components such that the problem can be solved in time $T$ on each connected component individually. In this case, we design an $\mathcal{O}^\star(2^n \cdot T)$ time algorithm for solving the problem on $G$ assuming that the valuations are integral and polynomially bounded (Theorem~\ref{thm:disjoint}) (without this assumption, the running time slightly worsens to $\mathcal{O}^\star(3^n \cdot T)$; Theorem~\ref{thm:disjoint}). {Notice that, when each component belongs to a polynomial-time solvable class (i.e., $T = n^{\mathcal{O}(1)}$ and when the number of connected components $r$ is $\Omega(n/\log n)$, our algorithm is faster than the aforementioned $\mathcal{O}(r!)$ algorithm.}  

A summary of our algorithmic results is given in Table~\ref{tab:summary}. 

\subsection{Related work}
As stated earlier, fair division and specifically house allocation problems have a rich and diverse history, and a relevant literature on the problem can be found in surveys e.g.,~\cite{DBLP:journals/ai/AmanatidisABFLMVW23,tapki2017brief}. More relevant to the current work are models that incorporate \emph{ locality of envy}. In the context of general envy-free allocation, this includes \cite{DBLP:journals/ai/EibenGHO23,DBLP:journals/ai/BredereckKN22}. Specifically, the former work considers treewidth of the social graph and the number of item types as a parameter; albeit crucially their algorithm has an \xp dependence on both (i.e., a running time of the form $\mathcal{O}^\star(n^{(\ell \cdot \mathrm{tw})})$), as opposed to our running time of $\mathcal{O}^\star(\ell^{\mathrm{tw}} \cdot n^{\mathcal{O}(\ell)})$, which is \fpt in $\mathrm{tw}$ when $\ell = \mathcal{O}(1)$. In the context of graphical house allocation,  Beyneir et al.~\shortcite{DBLP:journals/aamas/BeynierCGHLMW19} first studied the envy-free setting, and subsequent works ~\cite{DBLP:conf/atal/HosseiniPSVV23,DBLP:conf/atal/Hosseini0SVV24} consider the problem of minimizing envy. Other related works in this model \cite{MadathilMS25,DBLP:journals/corr/abs-2505-00296} consider different objectives, such as minimizing the maximum envy, or the number of envious agents.

\begin{table}[ht]
\centering
\footnotesize
\begin{tabularx}{0.6\textwidth}{|l|c|}

\cline{1-2}

\textbf{Setting} & \textbf{Time Complexity} \\
\cline{1-2}

\multicolumn{2}{|c|}{\textbf{Polynomial-Time Algorithms}} \\
\cline{1-2}
\makecell[l]{ Graphs with degree at most one (odd vertices) } &
polynomial time \\
\makecell[l]{  Graphs with degree at most one (even vertices) } &
1-additive approximation \\
\cline{1-2}

\multicolumn{2}{|c|}{\textbf{Parameterization by \(\ell\) Types of Houses}} \\
\cline{1-2}
\makecell[l]{treewidth \( \mathrm{tw} \)} &
\( \mathcal{O}^\star\left(\ell^{\mathrm{tw}}
\cdot \left(1 + \frac{n}{\ell}\right)^{2(\ell - 1)}\right) \) \\
\makecell[l]{Vertex cover number \( vc \)} &
\( \mathcal{O}^\star(\ell^{vc}) \) \\
\makecell[l]{Clique modulator \( k \)} &
\( \mathcal{O}^\star(\ell^{k}) \) \\
Complete bipartite graphs &
\( \mathcal{O}^\star\left( \left(1 + \frac{\min\{|L|,|R|\}}{\ell} \right)^\ell \right) \) \\
Split graphs &
\( \mathcal{O}^\star\left(\ell^{\min\{|C|, |I|\}} \right) \) \\
\cline{1-2}

\multicolumn{2}{|c|}{\textbf{Exact Algorithms}} \\
\cline{1-2}

Tree graphs &
\( \mathcal{O}^\star(2^n) \) \\
Disjoint union of graphs solvable in time \( T \) &
\( \mathcal{O}^\star(2^n \cdot T) \) \\
Cycles, paths, forests &
\( \mathcal{O}^\star(4^{n}) \) \\
\makecell[l]{\( f(n) \)-balanced separable graphs} &
\( \mathcal{O}^\star\left( 4^n \cdot 2^{\mathcal{O}(f(n)\log n)} \right) \) \\
Planar graphs &
\( \mathcal{O}^\star\left( 4^{n+o(n)} \right) \) \\
\cline{1-2}

\end{tabularx}
\caption{Summary of Algorithmic Results}
\label{tab:summary}

\end{table}

\section{Preliminaries}
\label{sec:prelims}

We begin by introducing the notation and basic definitions used throughout the paper,
and then formally define the optimization problem under consideration.

\begin{definition}[Envy Between Two Agents]
Let \( h_v \) and \( h_u \) be the houses assigned to agents \( v \) and \( u \), respectively.
The \emph{envy between agents \( v \) and \( u \)} is defined as
\begin{align*}
\envy(v, u, h_v, h_u) 
= \max\left\{ \val_v(h_u) - \val_v(h_v),\ 0 \right\} 
 + \max\left\{ \val_u(h_v) - \val_u(h_u),\ 0 \right\}.
\end{align*}
\end{definition}

% We measure fairness by minimizing the total envy over all social edges.
% This aggregate objective allows local trade-offs between neighboring agents and leads
% to a tractable optimization problem under various graph restrictions.

We now give the formal definition of \ghafull.

\begin{mdframed}
\textbf{Minimum Envy Graphical House Allocation} (\gha)
\\[0.5em]
\textbf{Input:}
\begin{itemize}
  \item An undirected graph \( G = (V, E) \), where each vertex represents an agent and each edge \( (v, u) \in E \) represents a social connection.
  \item A set of houses \( H \), with \( |H| = |V| \).
  \item For each agent \( v \in V \), a valuation function \( \val_v : H \to \mathbb{Q}_{\ge 0} \), where \( \val_v(h) \) denotes the value agent \( v \) assigns to house \( h \).
\end{itemize}

\textbf{Output:}  
An allocation \( \pi : V \to H \) that minimizes the total envy:
\[
\envy(\pi) \coloneqq \sum_{(v, u) \in E} \envy(v, u, \pi(v), \pi(u)).
\]
\end{mdframed}

Parameterized complexity provides a refined framework for analyzing computational
problems with respect to an additional parameter that captures structural properties
of the input.

\begin{definition}[Parameterized Problem]
A \emph{parameterized problem} is a decision or optimization problem whose input
is a pair \((I, k)\), where \(I\) is the main input and \(k \in \mathbb{N}\) is a
parameter. The size of the input is denoted as $|I|$.
\end{definition}

\begin{definition}[Fixed-Parameter Tractability (\fpt) and \xp]
A parameterized problem is \emph{fixed-parameter tractable (\fpt)} (resp.~$\xp$) with respect to a parameter \(k\) if it can be solved in time $f(k)\cdot |I|^{O(1)}$ (resp.~$|I|^{f(k)}$), for some computable function \(f\) depending only on \(k\).
\end{definition}

% \begin{definition}[\textsf{XP} Algorithms]
% A parameterized problem belongs is \textsf{XP} with respect to parameter \(k\) if it can be solved in time $|I|^{f(k)}$
% for some computable function \(f\).
% \end{definition}

\begin{definition}[W-hardness]
A parameterized problem is {\sf W[1]-hard} (resp.\ {\sf W[t]-hard}) if every problem
in the class {\textsf W[1]} (resp.\ {\sf W[t]}) can be reduced to it via a
parameterized reduction.
Unless \(\sf{FPT} = \sf{W[1]}\), no W[1]-hard problem admits an \fpt algorithm.
\end{definition}

We next recall the notion of \emph{treewidth}, a fundamental measure of graph sparsity that plays a central role in the design of dynamic programming algorithms on graphs. Treewidth-based techniques will be extensively used in Section~\ref{subsec:treewidth}.

\begin{definition}[Tree Decomposition and Treewidth]
A \emph{tree decomposition} of a graph \( G = (V, E) \) is a pair
\( (\mathcal{T}, \mathcal{B}) \), where \( \mathcal{T} \) is a tree and
\( \mathcal{B} = \{B_t \subseteq V \mid t \in V(\mathcal{T})\} \) is a collection of
vertex subsets called \emph{bags}, such that:
\begin{enumerate}
    \item For every vertex \( v \in V \), there exists a bag \( B_t \in \mathcal{B} \)
    with \( v \in B_t \).
    \item For every edge \( (u, v) \in E \), there exists a bag \( B_t \in \mathcal{B} \)
    such that \( \{u, v\} \subseteq B_t \).
    \item For every vertex \( v \in V \), the set
    \( \{ t \in V(\mathcal{T}) \mid v \in B_t \} \) induces a connected subtree of
    \( \mathcal{T} \).
\end{enumerate}
We root the tree $\mathcal{T}$ at an arbitrary vertex. For a node $t \in V(\mathcal{T})$ and use $\mathcal{T}_t$ to denote the subtree rooted at $t$ and $V_t \coloneqq \bigcup_{t \in V(\mathcal{T}_i)} B_t$. The \emph{width} of a tree decomposition is defined as
\( \max_{t} |B_t| - 1 \), and the \emph{treewidth} of \( G \), denoted
\( \mathrm{tw}(G) \), is the minimum width over all tree decompositions of \( G \).
\end{definition}

\begin{definition}[Nice Tree Decomposition] \label{def:nicetree}
A \emph{nice tree decomposition} of a graph \( G = (V, E) \) is a tree decomposition
\( (\mathcal{T}, \mathcal{B}) \) where \( \mathcal{T} \) is a rooted tree and every node
of \( \mathcal{T} \) is of one of the following four types:
\begin{itemize}
    \item \textbf{Leaf Node:} A node \( t \) with an empty bag, i.e., \( B_t = \emptyset \).

    \item \textbf{Introduce Node:} A node \( t \) with a single child \( t' \) such that
    \( B_t = B_{t'} \cup \{v\} \) for some vertex \( v \notin B_{t'} \).

    \item \textbf{Forget Node:} A node \( t \) with a single child \( t' \) such that
    \( B_t = B_{t'} \setminus \{v\} \) for some vertex \( v \in B_{t'} \).

    \item \textbf{Join Node:} A node \( t \) with two children \( t_1 \) and \( t_2 \),
    where \( B_t = B_{t_1} = B_{t_2} \).
\end{itemize}

It is well known that any tree decomposition of width \( k \) can be transformed,
in linear time, into an equivalent nice tree decomposition of width at most \( k \) and with \( \mathcal{O}(nk) \) nodes, where \( n = |V| \) ~\cite{cygan2015parameterized}. 
\end{definition}

We next recall several standard graph classes that appear in our complexity results.

\begin{definition}[Vertex Cover]
A \emph{vertex cover} of a graph \(G=(V,E)\) is a set \(C \subseteq V\) such that for
every edge \((u,v) \in E\), at least one of \(u\) or \(v\) belongs to \(C\).
\end{definition}

\begin{definition}[Split Graph]
A graph \(G=(V,E)\) is a \emph{split graph} if its vertex set can be partitioned into
two sets \(C\) and \(I\), where \(C\) is a clique and \(I\) is an
independent set.
\end{definition}

\begin{definition}[Bipartite Graph]
A graph \(G=(V,E)\) is called \emph{bipartite} if its vertex set \(V\) can be
partitioned into two disjoint sets \(L\) and \(R\) such that every edge
\(e \in E\) has exactly one endpoint in \(L\) and the other in \(R\).
The pair \((L,R)\) is called a \emph{bipartition} of \(G\).

A bipartite graph \(G=(L \cup R, E)\) is called \emph{complete bipartite} if
every vertex in \(L\) is adjacent to every vertex in \(R\).
\end{definition}

\section{Polynomial-Time optimal (approximation) Algorithm in Binary Case when the Social Graph has Maximum Degree $1$}

Let $\{G,H,\{ \val_v \}_{v\in V}\}$ be an instance of \gha, where
%$G=(V_G,E_G)$ is an undirected social graph whose vertices correspond
%to agents, and $H$ is the set of houses. 
for every agent $v\in V(G)$, the valuation function $\val_v : H \rightarrow \{0,1\}$. %satisfies $0 \le p < q$. 
We refer to such instances as \emph{binary instances}.
{We further assume that the maximum degree of $G$ is at most $1$.
Consequently, $G$ is a disjoint union of isolated vertices and isolated edges.
Formally, the vertex set $V(G)$ can be partitioned as $ V^{\text{iso}} \;\cup\; \{a_1, \ldots, a_{k}\}$ , where $V^{\text{iso}}$ is the set of isolated agents, and  for each $i\in\{1,\dots,k\}$, $a_i$ is a vertex of degree one. %the pair $\{a_i,a_{i+1}\}$ consists   of two distinct agents connected by an edge.

Let the edge set be
$E(G) = \bigcup_{i=1}^{\nicefrac{k}{2}} \{ (a_{2i-1},a_{2i}) \}$,
so that each agent is incident to at most one edge.
Whenever $(a_i,a_{i+1})\in E(G)$, we say that $a_i$ and $a_{i+1}$ are
\emph{adjacent agents}.
} Rather than assigning houses to each agent, we will assign a pair of house to an edge.

We design a polynomial-time algorithm that computes an \emph{almost envy-free} allocation for binary instances, when the social graph has degree at most one. In particular, at most one pair of adjacent agents may have a non-zero envy in the allocation returned by this algorithm.

Next, we prove an auxiliary lemma that lets us obtain an envy-free allocation for a pair of adjacent agents, as long as at least three houses are available.

\begin{restatable}{lemma} {lemthreehouse}
\label{obs:three-house-envy-resolution}
Let \( (a_i, a_j) \) be a pair of agents and \( \{h_k, h_{k+1}, h_{k+2}\} \) be any three houses. Then, two of these houses can always be allocated to \( a_i \) and \( a_j \) in an envy-free manner. 
\end{restatable}

\begin{proof} 
{
If two agents $a_i$ and $a_j$ are not adjacent in the social graph, then
they do not envy each other.
Hence, it suffices to consider adjacent agents.
By our notation, any adjacent pair is denoted as $(a_i,a_{i+1})$.
Consider the following two cases for assigning houses to agents
$a_i$ and $a_{i+1}$:
\begin{enumerate}
    \item The agents $a_i$ and $a_{i+1}$ can be assigned houses
    $h_k$ and $h_{k+1}$ such that the resulting tuple
    $(a_i,a_{i+1},h_k,h_{k+1})$ is a good or nice pair.
    In this case, the allocation is already envy-free.
    \item The agents $a_i$ and $a_{i+1}$ are assigned houses
    $h_k$ and $h_{k+1}$, respectively, and the tuple
    $(a_i,a_{i+1},h_k,h_{k+1})$ forms a bad pair.
    That is, both agents strictly prefer the same house over the other.
    In this case, an additional house can be used to resolve the
    resulting envy using the method described below.
\end{enumerate}
}

Without loss of generality, suppose
\[
\val_{a_i}(h_k) = \val_{a_{i+1}}(h_k) = 0, \quad
\val_{a_i}(h_{k+1}) = \val_{a_{i+1}}(h_{k+1}) = 1,
\]
The third house \( h_{k+2} \) has utility vector:
\[
\begin{aligned}
\val(h_{k+2})
&= \bigl[\val_{a_i}(h_{k+2}),\ \val_{a_{i+1}}(h_{k+2})\bigr] 
\in \{[0,0],\ [1,1],\ [0,1],\ [0,1]\}.
\end{aligned}
\]
We analyze each possible case:

\begin{description}
    \item[\textbf{Case 1:}] \( \val(h_{k+2}) = [0, 0] \) \\
    Swap \( h_{k+2} \) with \( h_{k+1} \). The agents are now assigned \( h_{k} \) and \( h_{k+2} \), both of which they value as \( 0 \). Thus, neither envies the other.

    \item[\textbf{Case 2:}] \( \val(h_{k+2}) = [1, 1] \) \\
    Swap \( h_{k+2} \) with \( h_k \). The agents are now assigned \( h_{k+2} \) and \( h_{k+1} \), both valued as \( 1 \), so neither envies the other.

    \item[\textbf{Case 3:}] \( \val(h_{k+2}) = [0, 1] \) \\
    Assign \( h_{k+2} \) to \( a_i \), and assign the house valued as \( 1 \) (either \( h_{k+1} \) or \( h_{k} \)) to \( a_{i+1} \). Now, \( a_i \) gets \( 0 \), and \( a_{i+1} \) gets \( 1 \), and since \( a_{i+1} \) sees no house better than their own in the other's allocation, envy is resolved.

    \item[\textbf{Case 4:}] \( \val(h_{k+2}) = [1, 0] \) \\
    Symmetric to Case 3: assign \( h_{k+2} \) to \( a_{i+1} \), and give \( a_i \) the house they value as \( 1 \). Again, both agents are content with their allocation.
\end{description}
\end{proof}

Armed with Lemma~\ref{obs:three-house-envy-resolution}, we prove the following result.

\begin{theorem} \label{thm:bivalued-warmup}
    There exists a polynomial-time algorithm that, given a binary instance $(G = (V, E), H, \LR{\val_v}_{v \in V})$ of \gha where $G$ has maximum degree at most $1$; and computes an allocation $\pi: V \to H$ with the following properties.
    \begin{enumerate}
        \item If all vertices in $G$ have degree exactly one, then all adjacent pairs of agents, except possibly one, are envy-free according to $\pi$.
        \item If there exists at least one isolated vertex in $G$, then $\pi$ is an envy-free allocation.
    \end{enumerate}
    %Let $G= (V, E)$ be a social graph where each vertex has degree exactly $1$. Then, there exists a polynomial-time algorithm that computes an allocation $\pi: V \to H$ such that for all adjacent pairs of agents, except possibly one, are envy-free according to $\pi$. Furthermore, if $G$ is a social graph of degree at most $1$ containing at least one isolated vertex, then the   
\end{theorem}
\begin{proof}
Since the given instance is binary, each agent $a \in V$ has a utility function $\val_a: H \to \{0, 1\}$. Recall that connected components of $G$ consist of either isolated vertices or disjoint edges.

\textbf{Step 1: Assign houses to isolated agents.}  
 For each $a \in V^{\text{iso}}$, assign any available house $h$ to $a$ that maximizes $\val_a(h)$. Since these agents have no neighbors, their assignments are trivially envy-free.

\textbf{Step 2: Assign houses to agent pairs.}  
Let $P \subseteq V$ be the set of remaining agents whose degree is $1$. %Since the graph is of degree at most 1, 
Clearly, $G[P]$ is a disjoint union of edges. Let $m = |P|/2$ be the number of such edges, and there are  $2m$ houses available for these edges (after assigning the isolated agents).

We allocate houses to these edges one-by-one. At each step, we choose any unassigned edge $(a_i, a_j)$ and two unassigned houses $(h_1, h_2)$, and try to assign them in a way that the pair is envy-free. Note that at any iteration except the last, we have two agents $(a_{i}, a_{i+1})$ connected by an edge, and at least four remaining houses (since this is not the last pair of agents). Therefore, by Lemma \ref{obs:three-house-envy-resolution}, we can find a pair of houses to be assigned to $(a_{i}, a_{i+1})$ that results in zero envy. Note that, in this procedure, only the last pair of agents may get assigned a pair of houses that results in envy $1$.

% This is always possible except potentially for the last pair, as explained next.

% We iteratively assign houses to agents in pairs formed by the edges in an arbitrary order. 
% \red{As long as there are more than two houses remaining (i.e., we are not assigning the last pair), we can always resolve any potential envy using the third available house. Specifically:

% The only configurations that necessarily result in envy are \emph{bad pairs}, and envy cannot be eliminated using only the two available houses.

% However, when at least three houses are available, we can always resolve such envy via  an extra house via Theorem \ref{obs:three-house-envy-resolution}.}

\textbf{Step 3: Use isolated agent's house to resolve envy in the last pair (if needed).}  
If $G$ contains at least one isolated vertex, then the house assigned to that vertex can serve as a third house in resolving any envy in the last edge. In this case, we always have at least three houses for the last edge.
\end{proof}

Due to Theorem~\ref{thm:bivalued-warmup} and the fact that the graph with odd number of vertices and maximum degree one has at least one isolated vertex, we obtain the following.
\begin{corollary}\label{cor:degree1}
     There exists a polynomial-time algorithm that, given a binary instance $(G = (V, E), H, \LR{\val_a}_{a \in V})$ \gha where $G$ has odd number of vertices and maximum degree at most one, computes an envy-free allocation, which is guaranteed to exist.
\end{corollary}
%We now describe the steps to eliminate envy in the final agent pair that may arise in the allocation produced by Theorem~\ref{thm:bivalued-warmup}.

% We conclude with the main result of this section, where we start with the almost-envy-free allocation given by Theorem~\ref{thm:bivalued-warmup}, and perform a careful case-analysis to obtain a minimum-envy allocation. 

% \begin{restatable}{theorem}{efbivalued}
% \label{thm:ef-bivalued}
% There exists a polynomial-time algorithm to compute a minimum-envy allocation in an instance of \gha where agents have bivalued utility functions, and the social graph has maximum degree at most one. Moreover, if the social graph either (i) contains at least one isolated vertex, or (ii) contains at least three edges, then the resulting allocation is envy-free.
% %There always exists an envy-free allocation in an instance of \gha where agents have bivalued utility functions, and the social graph has maximum degree at most one. Furthermore, such an allocation can be computed in polynomial time.
% \end{restatable}

The following example shows that there are instances with even number of agents that has envy one. Consider 
the following example. The social graph is a matching on four agents, with edges $(a_1,a_2)$ and $(a_3,a_4)$. The valuation profile is shown below.

\[
\begin{array}{c|cccc}
      & h_1 & h_2 & h_3 & h_4 \\
\hline
a_1   & 0   & 1   & 0   & 1 \\
a_2   & 0   & 1   & 0   & 1 \\
a_3   & 0   & 0   & 1   & 1 \\
a_4   & 0   & 0   & 1   & 1 \\
\end{array}
\]

In this example, envy of every allocation is at least one.
In fact, even for identical valuations, envy-free solution need not exist: suppose all but one house is valued $1$ by all the agents, one house is valued $0$.  

\section{\ghafull for \(\ell\) Types of Houses}

 \gha  is known to be \NP-hard in general~\cite{DBLP:conf/atal/HosseiniPSVV23}.
In this section, we study instances of \gha in which the set of houses is partitioned
into \(\ell\) distinct types.
Throughout this section, we treat \(\ell\) as a fixed constant and focus on the
parameterized complexity of the problem with respect to structural parameters of
the social graph.
Under this assumption, we design fixed-parameter tractable algorithms for \gha.
As corollaries of our results, we also obtain single-exponential time algorithms for several sparse graph classes, even in the extreme case where all houses are distinct, that is, when \(\ell = n\). We start by showing the following result, which shows that, the problem remains intractable for two house types, when the social graph is arbitrary. As discussed in Section~\ref{section:results}, this result is due to a reduction similar to~\cite[Theorem 3.2]{DBLP:conf/atal/HosseiniPSVV23}. 

\begin{restatable}{theorem}{thmlnphard}
\label{thm:me-gha-l2-nphard}
\gha is \NP-hard even when $\ell=2$.
\end{restatable}

\begin{proof}
We give a polynomial-time reduction from {\sc Minimum Bisection} problem, which is
\NP-hard.

\paragraph{Source problem definition.}
Let $G=(V,E)$ be an undirected graph with $|V|=n$, where $n$ is even.
The Minimum Bisection problem asks for a partition
$(V_0,V_1)$ of $V$ such that $|V_0|=|V_1|=n/2$ and the number of edges
crossing between $V_0$ and $V_1$ is minimized.

\paragraph{Construction.}
We construct an instance of ME-GHA as follows.

\textbf{Agents and social graph.}
For each vertex $v\in V$, create an agent $a_v$.
The social graph is same as the original graph.

\textbf{Houses.}
Create $n$ houses
\[
\mathcal H=\{h_1,\dots,h_n\},
\]
partitioned into two types:
\[
\mathcal H_0=\{h_1,\dots,h_{n/2}\}, \qquad
\mathcal H_1=\{h_{n/2+1},\dots,h_n\}.
\]

\textbf{Valuations.}
All agents have identical valuations
\[
v(h)=
\begin{cases}
0 & \text{if } h\in\mathcal H_0,\\
1 & \text{if } h\in\mathcal H_1.
\end{cases}
\]
Thus, $\ell=2$.

\paragraph{Forward direction.}
Let $(V_0,V_1)$ be a bisection of $G$.
We construct an allocation $\pi$ as follows:
assign each agent $a_v$ with $v\in V_1$ a distinct house from $\mathcal H_1$,
and each agent $a_v$ with $v\in V_0$ a distinct house from $\mathcal H_0$.
This is feasible since both sets have size $n/2$.

Consider an edge $\{u,v\}\in E$.
\begin{itemize}
    \item If $u,v\in V_0$ or $u,v\in V_1$, then both agents receive houses
    of equal value, so no envy arises.
    \item If $u\in V_0$ and $v\in V_1$ (or vice versa), then the agent in $V_0$
    envies the agent in $V_1$ by exactly one unit, while the reverse envy is zero.
\end{itemize}
Hence, each edge crossing $(V_0,V_1)$ contributes exactly one unit of envy.
Therefore,
\[
{\sf envy}(\pi)=|\{\{u,v\}\in E : u\in V_0,\,v\in V_1\}|.
\]

\paragraph{Backward direction.}
Let $\pi$ be any feasible allocation.
Since $|\mathcal H_0|=|\mathcal H_1|=n/2$, exactly $n/2$ agents receive value $1$.
Define
\[
V_1=\{a_v : v(\pi(a_v))=1\},\qquad
V_0=\{a_v : v(\pi(a_v))=0\}.
\]
Then $(V_0,V_1)$ is a bisection of $G$.

Using the same envy analysis as above, every edge with one endpoint in $V_0$
and the other in $V_1$ contributes exactly one unit of envy, and no other edges
contribute.
Thus,
\[
{\sf envy(\pi)}=|\{\{u,v\}\in E : u\in V_0,\,v\in V_1\}|.
\]

\paragraph{Correctness.}
The above establishes a one-to-one correspondence between bisections of $H$
and feasible allocations, preserving objective values.
Hence, an optimal allocation corresponds to a minimum bisection and vice versa.

Since the reduction is polynomial-time and Minimum Bisection is \NP-hard,
ME-GHA is \NP-hard even when $\ell=2$.
\end{proof}

This result shows that in order to design tractable algorithms parameterized by house types, we need to impose structural restrictions on the social graph. In the following subsections, we design parameterized algorithms for different restricted graph classes.

\subsection{Using Treewidth as a Parameter}
\label{subsec:treewidth}

Eiben et al.~\shortcite{DBLP:journals/ai/EibenGHO23} studied the problem of computing
envy-free allocations parameterized by the number of item types and the treewidth of the social graph. Their algorithm, however, exhibits an \xp dependence on the number of item types $\ell$ as well as the treewidth of the social graph. Notice that, when $\ell$ is a constant, their algorithm is still \xp in the treewidth. In contrast, our algorithm is \xp w.r.t.~$\ell$, but when $\ell$ is a constant, it is $\fpt$ w.r.t.~the treewidth of
the social graph.

We now state our main result for this section.

\begin{restatable}{theorem}{thmtw}
\label{thm:tw}
\gha can be solved in 
$\mathcal{O^*}\!\left(
n \cdot \ell^{\mathrm{tw}} \cdot \left(1 + \frac{n}{\ell}\right)^{2(\ell - 1)}
\right)$
time, where \(\mathrm{tw}\) denotes the treewidth of the social graph.
%In particular, under this assumption, \gha is fixed-parameter tractableparameterized by \(\mathrm{tw}\).
\end{restatable}

\begin{proof}
    We first fix notation and recall standard properties of nice tree decompositions.

\paragraph{Tree Decomposition.}
Let $G=(V,E)$ be the underlying graph.
A \emph{nice tree decomposition} of $G$ is a rooted tree $\mathcal{T}$ whose nodes are indexed by $i$, and where each node $i$ is associated with a bag $B_i \subseteq V$.  
For each node $i$, let $V_i$ denote the set of vertices appearing in the subtree rooted at $i$, i.e.,
\[
V_i = \bigcup_{j \text{ in the subtree of } i} B_j.
\]

We assume that $\mathcal{T}$ is a nice tree decomposition, meaning that every node is of one of the following types:

\begin{itemize}
    \item \textbf{Leaf node:} $B_i = \emptyset$.
    \item \textbf{Introduce node:} $i$ has one child $i'$ and $B_i = B_{i'} \cup \{v\}$ for some vertex $v$.
    \item \textbf{Forget node:} $i$ has one child $i'$ and $B_i = B_{i'} \setminus \{v\}$ for some vertex $v$.
    \item \textbf{Join node:} $i$ has two children $h,j$ and $B_i = B_h = B_j$.
\end{itemize}

A nice tree decomposition of width $\mathrm{tw}$ can be computed in
$f(\mathrm{tw}) \cdot n$ time using standard algorithms, and contains $O(n)$ nodes.  
This preprocessing step does not affect the final parameterized running time.

\medskip

\paragraph{Compatibility of allocations.}
Let $i$ be a node of the nice tree decomposition and let
$
(f,t_1,\dots,t_\ell)\in \mathcal{F}(i)
$
be a valid tuple.  
An allocation $\Psi : V_i \to [\ell]$ is \emph{compatible} with $(f,t_1,\dots,t_\ell)$ if:
\begin{itemize}
    \item $\Psi(v)=f(v)$ for every $v\in B_i$;
    \item for every $k\in[\ell]$, we have $|\Psi^{-1}(k)| = t_k$.
\end{itemize}

Let
$
O(i,f,t_1,\dots,t_\ell)
$
denote the minimum total envy achievable by any allocation on $G[V_i]$ that is compatible with $(f,t_1,\dots,t_\ell)$.

\paragraph{DP Table.}
We define a DP table $T[i, f, t_1, t_2, \dots, t_{\ell}]$, where the index \(i\) refers to a node of the tree decomposition. Here, $i, f$ are as defined above, and $(f, t_1, t_2, \ldots, t_\ell) \in \mathcal{F}(i)$. The entry stores the minimum envy of an allocation when for each $j \in [\ell]$, exactly $t_j$ houses of type $j$ are assigned to the vertices in the subtree rooted at $i$, and the vertices in the bag $B_i$ are assigned houses according to the function $f$.
% We define the DP table entry $ T[i,f,t_1,\dots,t_\ell]$
% to be the minimum total envy in $G[V_i]$ among all allocations compatible with $(f,t_1,\dots,t_\ell)$.
\\
We prove correctness by induction on the tree decomposition.
\\
\\
\textbf{Induction Hypothesis}

For every child node $i'$ of $i$ and every valid tuple
$(f',t_1',\dots,t_\ell') \in \mathcal{F}(i')$,
\[
T[i',f',t_1',\dots,t_\ell']
=
O(i',f',t_1',\dots,t_\ell').
\]

\paragraph{Goal}

Fix a valid tuple $(f,t_1,\dots,t_\ell)\in\mathcal{F}(i)$ and let
\[
T = T[i,f,t_1,\dots,t_\ell], \qquad
O = O(i,f,t_1,\dots,t_\ell).
\]

We show:
\[
T \ge O \quad\text{and}\quad T \le O.
\]

\hfill
\\
\textbf{Base Case (Leaf Node)}

If $i$ is a leaf, then $B_i=\emptyset$ and $V_i=\emptyset$.
There is exactly one valid tuple and no edges, hence no envy:
\[
T[i,f,t_1,\dots,t_\ell]=0=O(i,f,t_1,\dots,t_\ell).
\]

\hfill
\\
\textbf{Introduce Node}

Suppose $i$ is an introduce node with child $i'$ and
\[
B_i = B_{i'} \cup \{v\}.
\]

Let $(f,t_1,\dots,t_\ell)\in\mathcal{F}(i)$.
Define:
\[
f' = f|_{B_{i'}}, \qquad
t_k' =
\begin{cases}
t_k - 1 & \text{if } f(v)=k,\\
t_k & \text{otherwise}.
\end{cases}
\]

The DP recurrence is:
\[
\begin{aligned}
T[i,f,t_1,\dots,t_\ell]
=
T[i',f',t_1',\dots,t_\ell'] 
 +
\sum_{u\in N(v)\cap B_i}
\envy(v,u,f(v),f(u)).
\end{aligned}
\]

\paragraph{Part 1: showing \(T \ge O\)}

By the induction hypothesis, there exists a compatible allocation
$\Psi'$ for $(f',t'_1,\dots,t'_\ell)$ achieving
\[
\envy(\Psi') = T[i',f',t'_1,\dots,t'_\ell].
\]

Define $\Psi$ on $V_i$ by
\[
\Psi(w)=
\begin{cases}
\Psi'(w), & w\in V_{i'},\\
f(v), & w=v.
\end{cases}
\]

This allocation is compatible with $(f,t_1,\dots,t_\ell)$.
The only newly created envy edges are those incident to $v$ whose other endpoint lies in $B_i$.
Hence,
\[
\envy(\Psi)
=
\envy(\Psi')
+
\sum_{u\in N(v)\cap B_i}
\envy(v,u,f(v),f(u)).
\]

By definition of the DP transition,
\[
\envy(\Psi)=T[i,f,t_1,\dots,t_\ell].
\]
Since $O$ is the minimum achievable envy for this tuple, we obtain
\[
O \le T.
\]

\paragraph{Part 2:  showing  \(T \le O\)}

Assume for contradiction that $O < T$.
Let $\Psi^*$ be an optimal allocation compatible with $(f,t_1,\dots,t_\ell)$, so that
$\envy(\Psi^*) = O$.

Restrict $\Psi^*$ to $V_{i'}$ and denote the restriction by $\Psi'$.
Then $\Psi'$ is compatible with $(f',t'_1,\dots,t'_\ell)$.

The only envy terms lost when removing $v$ are those involving $v$, hence
\[
\begin{aligned}
\envy(\Psi')
&=
\envy(\Psi^*)
-
\sum_{u \in N(v)\cap B_i}
\envy(v,u,f(v),f(u)).
\end{aligned}
\]

Therefore,

Let $N_i(v) := N(v)\cap B_i$.
\[
\begin{aligned}
\envy(\Psi')
&<
T[i,f,t_1,\dots,t_\ell]
-
\sum_{u\in N_i(v)}
\envy(v,u,f(v),f(u)) \\
&=
T[i',f',t'_1,\dots,t'_\ell].
\end{aligned}
\]

This contradicts the induction hypothesis that
\[
T[i',f',t'_1,\dots,t'_\ell]
= O(i',f',t'_1,\dots,t'_\ell)
\]
is minimal. Hence $T \le O$.

\hfill
\\
\textbf{Forget Node}

Suppose $i$ is a forget node with child $i'$ and
\[
B_i = B_{i'} \setminus \{v\}.
\]

For any $(f,t_1,\dots,t_\ell)\in\mathcal{F}(i)$, define
\[
T[i,f,t_1,\dots,t_\ell]
=
\min_{\substack{f' : B_{i'}\to[\ell]\\ f'|_{B_i}=f}}
T[i',f',t_1,\dots,t_\ell].
\]

(Note that the counters $t_k$ do not change, since forgetting a vertex only removes it from the bag.)

\paragraph{Part 1:  showing \(T \ge O\)}

By the induction hypothesis, for each feasible extension $f'$ there exists an allocation $\Psi'$ compatible with $(f',t_1,\dots,t_\ell)$ attaining the value
$T[i',f',t_1,\dots,t_\ell]$.

Restricting $\Psi'$ to $V_i$ yields an allocation compatible with $(f,t_1,\dots,t_\ell)$, and forgetting a vertex does not change any envy value inside $G[V_i]$.

Thus, the DP value is achievable, and
\[
O \le T.
\]

\paragraph{Part 2:  showing \(T \le O\)}

Suppose $O < T$, and let $\Psi^*$ be an optimal allocation compatible with $(f,t_1,\dots,t_\ell)$.

Extend $\Psi^*$ arbitrarily to vertex $v$, obtaining an assignment $f'$ on $B_{i'}$.
This yields a compatible tuple $(f',t_1,\dots,t_\ell)$ at node $i'$.

Since removing or adding $v$ does not affect envy inside $G[V_i]$,
\[
\envy(\Psi^*) = \envy(\Psi').
\]

Hence,
\[
\envy(\Psi') < T[i,f,t_1,\dots,t_\ell]
\le T[i',f',t_1,\dots,t_\ell],
\]
contradicting the induction hypothesis.
Therefore $T \le O$.

\hfill
\\
\textbf{Join Node}

Let $i$ be a join node with children $h$ and $j$, where
\[
B_i = B_h = B_j.
\]

For a tuple $(f,t_1,\dots,t_\ell)\in\mathcal{F}(i)$, define
\[
b_k = |\{v\in B_i : f(v)=k\}|.
\]

The DP transition is:
\[
\begin{aligned}
T[i,f,\mathbf t]
=
\min_{\substack{\mathbf t^h,\mathbf t^j\\
t_k = t_k^h + t_k^j - b_k}}
\Big(
& T[h,f,\mathbf t^h]
+ T[j,f,\mathbf t^j] 
- \text{Envy}(B_i,f)
\Big),
\end{aligned}
\]
where
\[
\begin{aligned}
\text{Envy}(B_i,f)
=
\sum_{\{u,v\}\in E(G[B_i])}
\Big(
\envy(u,v,f(u),f(v)) 
+ \envy(v,u,f(v),f(u))
\Big).
\end{aligned}
\]

\paragraph{ Part 1:  showing  \(T \ge O\)}

By induction, there exist compatible allocations $\Psi_h,\Psi_j$ achieving
$T[h,f,\mathbf t^h]$ and $T[j,f,\mathbf t^j]$.

Because both agree with $f$ on $B_i$, they can be combined into an allocation
$\Psi_i$ on $V_i$.

The total envy satisfies:
\[
\envy(\Psi_i)
=
\envy(\Psi_h)
+
\envy(\Psi_j)
-
\text{Envy}(B_i,f),
\]
since edges internal to $B_i$ are counted twice.

Thus,
\[
\envy(\Psi_i)
=
T[h,f,\mathbf t^h]
+
T[j,f,\mathbf t^j]
-
\text{Envy}(B_i,f),
\]
which is exactly one candidate value in the minimization defining $T[i,f,\mathbf t]$.
Hence $O \le T$.

\paragraph{ Part 2:  showing \(T \le O\) }

Assume $T > O$ and let $\Psi^*$ be an optimal allocation compatible with $(f,\mathbf t)$.

Restrict $\Psi^*$ to $V_h$ and $V_j$, obtaining allocations $\Psi_h^*$ and $\Psi_j^*$.
These are compatible with tuples $(f,\mathbf t^h)$ and $(f,\mathbf t^j)$ satisfying
\[
t_k = t_k^h + t_k^j - b_k.
\]

By induction,
\[
\envy(\Psi_h^*) \ge T[h,f,\mathbf t^h], \qquad
\envy(\Psi_j^*) \ge T[j,f,\mathbf t^j].
\]

Since envy on edges inside $B_i$ is counted twice,
\[
\envy(\Psi^*)
=
\envy(\Psi_h^*)
+
\envy(\Psi_j^*)
-
\text{Envy}(B_i,f).
\]

Combining inequalities,
\[
\envy(\Psi^*)
\ge
T[h,f,\mathbf t^h]
+
T[j,f,\mathbf t^j]
-
\text{Envy}(B_i,f)
\ge
T[i,f,\mathbf t],
\]
contradicting $T > O$. Thus, $T \le O$.

\paragraph{Conclusion.}
For every node $i$ and every valid tuple $(f,t_1,\dots,t_\ell)$,
\[
T[i,f,t_1,\dots,t_\ell] = O(i,f,t_1,\dots,t_\ell).
\]
This completes the proof.

\paragraph{Running Time.}
A nice tree decomposition of width $\mathrm{tw}$ can be computed in
$f(\mathrm{tw})\cdot n$ time.

For each bag, the number of assignments of $f$ is at most $\ell^{\mathrm{tw}+1}$.
The number of feasible vectors $(t_1,\dots,t_\ell)$ is
\[
\mathcal{O}\!\left(\left(1+\frac{n}{\ell}\right)^\ell\right).
\]

Join nodes require iterating over pairs of such vectors, giving total time
\[
\mathcal{O}^*\!\left(
\ell^{\mathrm{tw}}
\cdot
\left(1+\frac{n}{\ell}\right)^{2(\ell-1)}
\right).
\]

\end{proof}

Due to Theorem~\ref{thm:tw} and the fact that $\ell \leq n$, we have the following result.

\begin{corollary}
\label{cor:constwidth}
When the treewidth of the social graph is constant, \gha can be solved in
\(\mathcal{O}^\star(4^n)\) time; in particular, this includes cycles, paths, and forests.
\end{corollary}

\begin{definition}[\(f(n)\)-Balanced-Separable Graph Classes]
For a function $f: \mathbb{N} \to \mathbb{N}$,
A graph class \(\mathcal{G}\) is said to be \emph{\(f(n)\)-balanced-separable} if for every graph $G \in \mathcal{G}$ with $n$ vertices, (i) all induced subgraphs of $G$ also belong to \(\mathcal{G}\), and (ii) there exists a partition
\((S,V_1,V_2)\) of \(V(G)\) such that:
\begin{itemize}
    \item \(|S| \le f(n)\),
    \item \(|V_1|, |V_2| \le \frac{2n}{3}\),
    \item there are no edges between \(V_1\) and \(V_2\).
\end{itemize}
\end{definition}

Planar graphs are a classical example of an \(\mathcal{O}(\sqrt{n})\)-balanced-separable graph
class.
Moreover, it is well known that if a graph belongs to an \(f(n)\)-balanced-separable
class, then its treewidth is bounded by \(\mathcal{O}(f(n))\)~\cite{cygan2015parameterized}.
This yields the following corollary.

\begin{corollary}
\label{cor:balsep}
If the social graph belongs to an \(f(n)\)-balanced-separable class, then \gha can be
solved in
\(\mathcal{O}^\star\!\left(4^n \cdot 2^{\mathcal{O}(f(n)\log n)}\right)\) time.
In particular, if the social graph is planar, then \gha can be solved in
\(\mathcal{O}^\star(4^{n+o(n)})\) time.
\end{corollary}

\subsection{\gha on Graphs with Small Structural Modulators}

\paragraph{Note on Minimum Weight Perfect Matching in bipartite graph}
Let \(G=(V,E)\) be an undirected graph. A \emph{perfect matching} is a subset \(M\subseteq E\) such that every vertex \(v\in V\) is incident to exactly one edge in \(M\).  Equivalently, \(M\) covers all vertices and \(\lvert M\rvert = \lvert V\rvert/2\)

In the \emph{minimum weight perfect matching} (MWPM) problem, each edge \(e\in E\) has a real weight \(c(e)\).  The goal is to find a perfect matching \(M^*\) minimizing the total weight
\[
  c(M)\;=\;\sum_{e\in M}c(e).
\]
Formally:
\begin{itemize}
  \item \textbf{Input:} A weighted graph \((V,E,c)\) with weight function \(c:E\to\mathbb{R}\).
  \item \textbf{Output:} A perfect matching \(M^*\subseteq E\) such that 
    \(\displaystyle M^* = \arg\min_{M\ \text{perfect matching}} \sum_{e\in M}c(e).\)
\end{itemize}

% In bipartite graphs, the MWPM reduces to the assignment problem. \todo{what's the need of this?}  
The \emph{Hungarian algorithm} ~\cite{DBLP:books/daglib/p/Kuhn10} solves it in \(\mathcal{O}(n^3)\) time. 

We begin with defining modulators. 

\paragraph{Modulators.}
Let \(G=(V,E)\) be a graph and let \(\mathcal{G}\) be a graph class.
A set \(S\subseteq V\) is called a \emph{modulator to \(\mathcal{G}\)} if
\(G-S \in \mathcal{G}\).
If \(|S|\le k\), we say that \(S\) is a \emph{small modulator} (of size \(k\)).

Note that vertex cover is a modulator to edgeless graphs. Modulator to \(\mathcal{G}\) can also be thought of as a distance from the graph class \(\mathcal{G}\). In the parameterized complexity, the distance from the polynomial-time solvable instances is a natural and well-studied parameter~\cite{10.1007/978-3-031-06678-8_1,10.1007/978-3-540-28639-4_15}. In this paper, we consider two simple graph classes that are polynomial-time solvable: edgeless graphs and complete graphs.

%We next study \gha under structural restrictions that admit small  modulators. We begin with vertex cover and clique modulator, which yield a
%unified algorithmic framework. Split graphs will then follow as a special case.

\subsubsection{Vertex Cover as a Parameter (Modulator to edgeless graphs)}

Let \(X\subseteq V\) be a vertex cover of \(G\). Then the graph \(G-X\) is edgeless. %This structural property allows us to isolate all edges inside
%the vertex cover, while the remaining vertices are mutually non-adjacent.

\begin{restatable}{theorem}{thmghavc}
\label{thm:gha-vc-typed}
\gha can be solved in
\(\mathcal{O}^\star(\ell^{\mathrm{vc}})\) time, where  \(\mathrm{vc}\) is the size of a vertex cover. %It admits \fpt parameterised by $vc$ when the number of house types \(\ell\) is treated as a constant
\end{restatable}

\paragraph{Algorithmic idea.}
Let $X=\{x_1,\ldots,x_k\}$. For each agents in $X$, we guess the type of the house assigned to them. Let $t$ be the type of the house assigned to $x_1$. Then, we allocate a house of type $t$ to $x_1$. Similarly, we allocate houses to all the agents in $X$. Let $\{h_1,\ldots,h_k\}$ be the houses allocated to agents in $X$ respectively, and $H'$ be the set of remaining houses. Then, we create a complete weighted bipartite graph between $I=V \setminus X$ and $H'$, where the weight of an edge $uh$, where $u\in I, v\in H'$, is $\sum_{x_i\in N_X(u)}\envy(u,x_i,h,h_i)$, i.e., the total envy between agents $u$ and each of its neighbors in $X$, if the house $h$ is allocated to $u$. Then, we find minimum-weight perfect matching in this graph in polynomial time~\cite{DBLP:books/daglib/p/Kuhn10}. This corresponds into an allocation of houses to agents in $G-X$, which can be extended to a complete allocation by assigning the house $h_i$ to each $x_i \in X$. We return an allocation that minimizes envy. The running time is dominated by the number of guesses for the agents in $X$, which is $\ell^{vc}$. 

\begin{proof}
Let $(G, H, \LR{\val_v}_{v \in V})$ be an instance of \gha. Let $X=\{x_1,\ldots,x_{vc}\}$ be a vertex cover of $G$ and $T$ be set of the types of houses in $H$. Let $\gamma\colon X \rightarrow T$ be an assignment of agents to types of houses and $\Gamma$ be the set of all such possible assignments. Note that $|\Gamma| \leq |T|^{|X|}=\ell^{vc}$. For each assignment, $\gamma \in \Gamma$, we allocate a house of type $\gamma(x_i)$ to $x_i$, where $i\in [vc]$. Let $\{h_1,\ldots,h_{vc}\}$ be the houses allocated to agents in $X$ and $H_\gamma$ be the set of remaining houses. Let $I=V(G-X)$. We create a complete weighted bipartite graph $B_\gamma=(I\cup H_\gamma)$, where the weight of an edge $uh$, where $u\in I, h\in H_\gamma$, is $\sum_{x_i\in N_X(u)}\envy(u,x_i,h,h_i)$. Then, we find minimum weight matching, $M_\gamma$, of the graph $B_\gamma$ using the algorithm in ~\cite{DBLP:books/daglib/p/Kuhn10}. We construct an allocation $\eta_\gamma$ as follows: for $x_i\in X$, $\eta_\gamma(x_i)=h_i$ and for $x\in I$, $\eta_\gamma(x)=M_\gamma(x)$. We return an allocation $\eta = \arg\min_{\gamma \in \Gamma} \envy(\eta_\gamma)$.

Next, we prove the correctness of the algorithm. Let $\mu$ be an optimal allocation. Let $T_\mu$ be the  set of type of houses assigned to agents in $X$. Consider allocation $\gamma \in \Gamma$ such that $\gamma(X)=T_\gamma$. Note that envy between the agents in $X$  is same in both the allocations $\mu$ and $\eta_\gamma$ and there is no envy between the agents in $I$. $\gamma$. Suppose that $\envy(\mu)<\envy(\eta_\gamma)$. Then, $$\sum_{uv\colon u\in I, v\in X} \envy(u,v,\mu) < \sum_{uv\colon u\in I, v\in X} \envy(u,v,\eta_\gamma)$$ 

Let $H_\mu$ be the set of houses assigned to agents in $I$ by the allocation $\mu$. Note that the number of houses of a type is same in both $H_\mu$ and $H_\gamma$ due to the choice of $\gamma$. Hence, we construct a matching $M$ in $B_\gamma$ as follows: for an agent $u\in I$, $M(u)$ is a house in $H_\gamma$ whose type is same as the type of the house $\mu(u)$. Since $\sum_{uv\colon u\in I, v\in X} \envy(u,v,\mu) < \sum_{uv\colon u\in I, v\in X} \envy(u,v,\eta_\gamma)$, it follows that the weight of matching $M$ is smaller than the weight of matching $M_\gamma$, a contradiction to the fact that $M_\gamma$ is a minimum weight matching. 

Since the minimum weight matching can be computed in polynomial time~\cite{DBLP:books/daglib/p/Kuhn10} and size of $|\Gamma|$ is $\ell^{vc}$, the algorithm runs in ${\mathcal O}^\star(\ell^{vc})$ time.

\paragraph{Computing a Vertex Cover.} Computing a minimum vertex cover is \np-hard, but the problem is {\sf FPT} when parameterized by the vertex cover size \( vc \), with algorithms running in \( O^*(1.25284^{vc}) \) time~\cite{DBLP:conf/stacs/0001N24}
\end{proof}

% We branch on all possible assignments of house types to the vertices in the
% vertex cover \(X\). Since \(|X|=\mathrm{vc}\), this yields at most
% \(\ell^{\mathrm{vc}}\) branches.

% Once the assignment on \(X\) is fixed, the remaining vertices form an independent
% set. The residual problem of assigning houses to these vertices in order to
% minimize total envy can be reduced to a minimum-weight perfect matching instance
% in a bipartite graph, where edge weights encode envy contributions with respect
% to already fixed neighbors.

% Since minimum-weight perfect matching can be solved in polynomial time, each
% branch can be processed efficiently, yielding the stated running time.

Due to Theorem~\ref{thm:gha-vc-typed}, we have the following corollary, since the size of the vertex cover of a bipartite graph \(G=(L\uplus R,E)\) is at most \(\min\{|L|,|R|\}\). 

\begin{corollary}\label{cor:bipartite}
    For a bipartite graph \(G=(L\uplus R,E)\), \gha can be solved in  \(\mathcal{O}^\star(\ell^{\min\{|L|,|R|\}})\) time.
\end{corollary}

Recall that \gha is \NP-hard for bipartite graphs.
% \paragraph{Bipartite graphs.}
% In any bipartite graph \(G=(L\cup R,E)\), one of the partitions forms a vertex
% cover. Hence, if \(k=\min(|L|,|R|)\), then \(G\) admits a vertex cover of size at most \(k\). Applying Theorem~\ref{thm:gha-vc-typed}, we obtain that \gha on
% bipartite graphs can be solved in \(\mathcal{O}^\star(\ell^k)\) time when \(\ell\) is
% constant.

%-------------------------------------------------
\subsubsection{Clique Modulator as a Parameter}
%-------------------------------------------------

% We now generalize the above approach using clique modulator.

A \emph{clique modulator} is a set of vertices whose removal turns the graph into
a complete graph. Let \(X\) be such a modulator of size \(k\).

\begin{restatable}{theorem}{thmcliquemod}
\label{thm:gha-clique-modulator}
\gha can be solved in
\(\mathcal{O}^\star(\ell^k)\) time, where $k$ is the size of a clique modulator.
\end{restatable}

\paragraph{Algorithmic Idea.}
The idea is similar to previous algorithm. 
For each agents in $X$, we guess the type of the house assigned to them, and assign houses to them as earlier. Let $H'$ is the set of remaining houses. We create a complete weighted bipartite graph on $I,V\setminus X$ and $H'$. In this case, agents in $G-X$ also envy each other. Thus, we set weight on edges as follows: the weight of an edge $uh$, where $u\in I, h\in H'$, is $\sum_{x_i\in N_X(u)}\envy(u,x_i,h,h_i)+\sum_{h\in H'}\max\{\alpha_u(h)-\alpha_u(v),0\}$.

\begin{proof}
The algorithm is same as the previous algorithm (Theorem~\ref{thm:gha-vc-typed}), the only difference is the weight on the edges in the bipartite graph. The weight of an edge $uh$, where $u\in I, h\in H'$, is $\sum_{x_i\in N_X(u)}\envy(u,x_i,h,h_i)+\sum_{h\in H'}\max\{\alpha_u(h)-\alpha_u(v),0\}$. The correctness proof and running time analysis is also same as earlier.

\paragraph{Computing a Clique Modulator.}
 Computing a minimum clique modulator is NP-hard, but the problem is fixed-parameter tractable (FPT) when parameterized by the modulator size \( k \), with algorithms running in \( O^*(2^k) \) time~\cite{DBLP:journals/siamdm/GutinMOW21}.
\end{proof}

% After fixing these assignments, the remaining graph is a disjoint union of
% cliques. For each clique, the residual allocation problem can be solved
% independently using matching-based techniques.

% Combining the solutions for all cliques gives a feasible assignment for the
% entire graph. Since each branch is solvable in polynomial time, the total running
% time follows.

% A complete proof is deferred to the supplementary material.

% \textbf{Split graph.} Recall that in a split graph, the vertex set $V$ is partitioned into a clique $C$ and an independent set $I$. In other words, $I$ is a clique modulator, and $C$ is a vertex cover in the graph. Therefore, by comparing the sizes of $C$ and $I$, and running the faster algorithm out of

Due to Theorem~\ref{thm:gha-vc-typed} and \ref{thm:gha-clique-modulator}, we obtain the following result. 
% \subsubsection{Special case of Split Graphs}

% We now derive an algorithm for split graphs as a direct consequence of the
% previous parameterized results.

\begin{restatable}{corollary}{thmsplit}
\label{thm:splitgraphs}
For a split graph $G=(C \uplus I, E)$, \gha can be solved in
\(\mathcal{O}^\star(\ell^{\min\{|C|,|I|\}})\) time.
\end{restatable}

%The above corollary is sue to the fact that the size of vertex cover is at most $|I|$ and size of clique modulator is at most $|C|$.

% \begin{proof}
% Let \(G=(V,E)\) be a split graph with partition \(V=C\cup I\), where
% \(C\) is a clique and \(I\) is an independent set.
% Let \(k=\min(|C|,|I|)\).

% We distinguish two cases.

% \medskip
% \noindent
% \textbf{Case 1: \(|C|\le |I|\).}
% Since \(I\) is an independent set, every edge of \(G\) has at least one endpoint
% in \(C\). Hence, \(C\) forms a vertex cover of \(G\) of size \(|C|\).
% By Theorem~\ref{thm:gha-vc-typed}, \gha can be solved in
% \(\mathcal{O}^\star(\ell^{|C|})\) time.

% \medskip
% \noindent
% \textbf{Case 2: \(|I|< |C|\).}
% Since \(C\) is a clique, removing \(I\) leaves a single clique.
% Therefore, \(I\) is a clique modulator of size \(|I|\).
% By Theorem~\ref{thm:gha-clique-modulator}, \gha can be solved in
% \(\mathcal{O}^\star(\ell^{|I|})\) time.

% \medskip
% In both cases, the running time is bounded by
% \(\mathcal{O}^\star(\ell^k)\), where \(k=\min(|C|,|I|)\).
% Hence, the stated bound follows.
% \end{proof}

\subsection{\gha on Complete Bipartite Graphs}

Complete bipartite graph is polynomial-time solvable for identical valuations~\cite{DBLP:conf/atal/HosseiniPSVV23}. We show that it is also polynomial-time solvable for constant number of house types. However, the complexity remains open, in general.
% We begin by establishing that \gha remains computationally hard even on very restricted
% bipartite graph structures, provided isolated vertices are allowed.

% \begin{restatable}{theorem}{thmnpbip}
% \label{thm:np-bip}
% Let \(\varepsilon \in (0,1)\) be any constant.
% Consider an instance with \(n\) agents and \(n\) houses.
% The \gha problem is \NP-hard even when the social graph is a disjoint union of a complete
% bipartite graph and isolated vertices, where one side of the bipartition has size at most
% \(n^{\varepsilon}\), and each house is preferred by at most four agents.
% \end{restatable}

% Due to space constraints, the proof of Theorem~\ref{thm:np-bip} is deferred to the
% supplementary material.

% The above result shows that allowing even a small number of isolated vertices suffices to make the problem computationally intractable. Next, we propose a polynomial-time algorithm for this graph class when the number of house types is constant.
%We now see the case where the houses are partitioned into \(\ell\) types.

\begin{restatable}{theorem}{thmcbipindv}
\label{thm:gha-complete-bipartite-indv}
 For a complete bipartite graph $G=(L\uplus R)$, \gha 
can be solved in $\mathcal{O}^\star\!\left(\left(1 + \frac{k}{\ell}\right)^{\ell}\right)$
time, where $k = \min\{|L|,|R|\}$. %In particular, when \(\ell\) is treated as a constant, the problem admits a polynomial-time algorithm.
\end{restatable}

\begin{proof}
Let $G=(L\cup R,E)$ be a complete bipartite graph.
Let
\[
k=\min\{|L|,|R|\},
\]
and without loss of generality assume $|L|=k$.
Then $|R|=n-k$.
There are $n$ houses partitioned into $\ell$ types.

\medskip
\noindent
\textbf{Algorithm.}
We enumerate all tuples $(t_1,\ldots,t_\ell)$ such that
\[
\sum_{i=1}^{\ell} t_i = k,
\]
where $t_i$ denotes the number of houses of type $i$ assigned to agents in $L$.
For each such tuple, we enumerate all assignments of houses to agents in $L$
that respect these counts.
For every assignment, we compute the total envy contributed by edges incident
to vertices in $L$.

Since $G$ is complete bipartite, every edge has one endpoint in $L$ and one in
$R$.
The remaining houses are assigned arbitrarily to agents in $R$.

\medskip
\noindent
\textbf{Correctness.}
We show that the algorithm returns an optimal allocation.

Consider any optimal allocation $\pi^*$.
Let $(t_1^*,\ldots,t_\ell^*)$ be the number of houses of each type assigned to
agents in $L$ under $\pi^*$.
By construction,
\[
\sum_{i=1}^{\ell} t_i^* = |L| = k,
\]
so this tuple is enumerated by the algorithm.
Moreover, the restriction of $\pi^*$ to agents in $L$ is one of the assignments
considered for this tuple.

Since all edges are between $L$ and $R$, the total envy depends only on which
houses are assigned to agents in $L$.
Therefore, when the algorithm evaluates this assignment, it computes exactly
the same envy value as in $\pi^*$.
Hence, the algorithm will consider an allocation with optimal envy and will
output an allocation of minimum possible envy.

\medskip
\noindent
\textbf{Running time.}
The number of tuples $(t_1,\ldots,t_\ell)$ satisfying
$\sum_{i=1}^{\ell} t_i=k$ is
\[
O\!\left(\left(1+\frac{k}{\ell}\right)^\ell\right).
\]
For each tuple, the number of consistent assignments is polynomial in $n$.
Thus, the total running time is
\[
O\!\left(\left(1+\frac{k}{\ell}\right)^\ell\cdot \poly(n)\right),
\]
which is polynomial in $n$ when $\ell$ is treated as a constant.

\medskip
\noindent
\textbf{Conclusion.}
Therefore, \gha\ is solvable in polynomial time on complete bipartite graphs
(augmented with isolated vertices) when the number of house types $\ell$ is
constant.
\end{proof}

% In the case where the social graph is a \emph{pure} complete bipartite graph,
% we do not need to perform a matching step. \todo{why are you writing it here? There is no proof, so it does not have any context.}
% After guessing an assignment for
% \( \min(|L|,|R|) \) agents, all remaining agents lie on the opposite side of the bipartition and are adjacent to every guessed agent.
% Hence, the remaining houses can be assigned to them arbitrarily, since this does not affect the envy structure.\todo{this whole para is not useful.}

% \todo[inline]{what's the need of following corollary.}

% \begin{restatable}{corollary}{thmcbip}
% \label{thm:gha-complete-bipartite}
% When the social graph is a complete bipartite graph and the houses are partitioned into
% \(\ell\) types, the \gha problem can be solved in $\mathcal{O}^\star\!\left(\left(1 + \frac{k}{\ell}\right)^{\ell}\right)$ time where $k = min(|L|,|R|)$. In particular, when \(\ell\) is treated as a constant, the problem admits a
% polynomial-time algorithm.
% \end{restatable}

\section{Exact Algorithms for \ghafull}

It is known that the \gha problem remains \NP-hard even under identical valuations on
trees~\cite{DBLP:conf/atal/Hosseini0SVV24} and on disjoint unions of graphs~\cite{DBLP:conf/atal/HosseiniPSVV23}.
In this section, we consider the general case of non-identical valuations and present
exact exponential-time algorithms.
Specifically, for instances in which agents’ valuation functions take
polynomially bounded integer values, we design algorithms running in
\(\mathcal{O}^\star(2^n)\) time for trees, and in \(\mathcal{O}^\star(2^n \cdot T)\) time
for disjoint unions of graphs, where \(T\) denotes the time required to solve
\gha on a given graph class.

\subsection{\gha on Trees}

\begin{restatable}{theorem}{treethm}
\label{thm:tree}
When the social graph \(G\) is a tree, \gha can be solved in
\(\mathcal{O}^\star(3^n)\) time.
Moreover, if the agents’ valuation functions take polynomially bounded integer values,
the problem can be solved in \(\mathcal{O}^\star(2^n)\) time.
\end{restatable}

\begin{proof}

Let the underlying structure be a rooted tree.
For a vertex $v$, let its children be denoted by $u_1,\dots,u_k$.
Let $H$ be the set of houses.

We define the following quantities.

\begin{itemize}
    \item $O(v,h_v,S)$ denotes the minimum total envy achievable in the subtree rooted at $v$, assuming that vertex $v$ is assigned house $h_v$ and the set of houses used in this subtree is exactly $S \subseteq H$.
    \item $Q(v,i,h_v,Z)$ denotes the minimum total envy in the union of the subtrees rooted at the first $i$ children $u_1,\dots,u_i$, assuming that $v$ receives house $h_v$, and the set of houses used by these $i$ subtrees is exactly $Z \subseteq H\setminus\{h_v\}$.
\end{itemize}

We compute these values using the dynamic programming tables
\[
T[v,h_v,S] := O(v,h_v,S),
\qquad
P[v,i,h_v,Z] := Q(v,i,h_v,Z).
\]

\paragraph{Base Case.}
If $v$ is a leaf, then its subtree consists only of $v$. Hence,
\[
T[v,h_v,S] =
\begin{cases}
0, & \text{if } S = \{h_v\},\\
\infty, & \text{otherwise}.
\end{cases}
\]
This is correct since no envy can arise in a single-vertex subtree.

\paragraph{Recurrence Definitions.}

Let $u_1,\dots,u_k$ be the children of $v$.

\paragraph{Initialization for the first child.}
For the first child $u_1$, define
\[
\begin{aligned}
P[v,1,h_v,Z]
=
\min_{h_1 \in Z}
\Big(
T[u_1,h_1,Z]
+
\envy(v,u_1,h_v,h_1)
\Big).
\end{aligned}
\]

This expression accounts for the optimal envy within the subtree of $u_1$ and the envy created between $v$ and $u_1$.

\paragraph{General recurrence.}
For $i \ge 2$, define
\[
\begin{aligned}
P[v,i,h_v,Z]
&=
\min_{\substack{X \cup Y = Z\\ X \cap Y = \emptyset}}
\Big(
P[v,i-1,h_v,X] \\
&\quad +
\min_{h_i \in Y}
\big(
T[u_i,h_i,Y]
+
\envy(v,u_i,h_v,h_i)
\big)
\Big).
\end{aligned}
\]

This recurrence distributes the set $Z$ of available houses between the first $i-1$ children and the $i$-th child, and adds the envy created between $v$ and that child.

\paragraph{Final combination at node $v$.}
Once all children are processed, define
\[
T[v,h_v,S]
=
\begin{cases}
P[v,k,h_v,S\setminus\{h_v\}], & \text{if } h_v \in S,\\
\infty, & \text{otherwise}.
\end{cases}
\]

This ensures that $v$ receives house $h_v$, while the remaining houses are distributed among its children.

\paragraph{Correctness Proof.}

We prove correctness by induction on the depth of the tree.

\paragraph{Induction Hypothesis.}
For every node $u$ strictly below $v$, and for all valid pairs $(h_u,S)$, the values
\[
T[u,h_u,S] \quad \text{and} \quad P[u,i,h_u,Z]
\]
correctly represent the minimum total envy in their respective subproblems.

\paragraph{Goal.}
We show that for node $v$ and every valid pair $(h_v,S)$,
\[
T[v,h_v,S] = O(v,h_v,S).
\]

\paragraph{Correctness of $P[v,1,h_v,Z]$.}

\textbf{($\ge$ direction).}
For any choice of $h_1 \in Z$, the quantity
\[
T[u_1,h_1,Z] + \envy(v,u_1,h_v,h_1)
\]
represents the envy incurred inside the subtree of $u_1$ together with the envy between $v$ and $u_1$.
By the induction hypothesis, $T[u_1,h_1,Z]$ is minimal.
Hence no feasible solution can have smaller cost, implying
\[
P[v,1,h_v,Z] \ge Q(v,1,h_v,Z).
\]

\textbf{($\le$ direction).}
Let an optimal solution assign house $h_1^*$ to $u_1$.
By the induction hypothesis,
\[
T[u_1,h_1^*,Z] = O(u_1,h_1^*,Z).
\]
The recurrence explicitly considers this choice, hence
\[
P[v,1,h_v,Z] \le Q(v,1,h_v,Z).
\]

Thus equality holds.

\paragraph{Correctness of $P[v,i,h_v,Z]$ for $i \ge 2$.}

\textbf{($\ge$ direction).}
The recurrence enumerates all partitions $Z = X \cup Y$ with $X \cap Y = \emptyset$.
For each such partition, it combines:
\begin{itemize}
    \item the optimal envy for children $u_1,\dots,u_{i-1}$ using $X$, and
    \item the optimal envy for child $u_i$ using $Y$, together with the envy between $v$ and $u_i$.
\end{itemize}
By the induction hypothesis, both components are minimal for their subproblems.
Thus,
\[
P[v,i,h_v,Z] \ge Q(v,i,h_v,Z).
\]

\textbf{($\le$ direction).}
Let an optimal solution for $Q(v,i,h_v,Z)$ assign disjoint subsets
$Z_1,\dots,Z_i$ to children $u_1,\dots,u_i$, with $Z_i$ assigned to $u_i$.
Let $h_i^* \in Z_i$ be the house assigned to $u_i$.

By the induction hypothesis,
\[
P[v,i-1,h_v,Z\setminus Z_i]
=
Q(v,i-1,h_v,Z\setminus Z_i),
\]
and
\[
T[u_i,h_i^*,Z_i] = O(u_i,h_i^*,Z_i).
\]

Since this exact combination appears in the minimization defining $P[v,i,h_v,Z]$, we obtain
\[
P[v,i,h_v,Z] \le Q(v,i,h_v,Z).
\]

\paragraph{Correctness of $T[v,h_v,S]$.}

By definition,
\[
T[v,h_v,S] = P[v,k,h_v,S\setminus\{h_v\}].
\]

Since $P[v,k,\cdot]$ correctly computes the minimum envy over all valid allocations among the children, this value equals the minimum total envy achievable in the subtree rooted at $v$ when $v$ is assigned $h_v$.

Therefore,
\[
T[v,h_v,S] = O(v,h_v,S).
\]

\paragraph{Conclusion.}
By induction on the depth of the tree, all dynamic programming entries satisfy
\[
T[v,h_v,S] = O(v,h_v,S),
\]
which completes the proof.

\paragraph{Time Complexity.}
For each vertex $v$, the table $T[v,h_v,S]$ has $O(n2^n)$ entries.
The auxiliary table $P[v,i,h_v,Z]$ considers all partitions of a set, leading to
\[
\sum_{k=0}^n \binom{n}{k} 2^k = 3^n
\]
total combinations.
Hence the total running time is
\[
\mathcal{O}^*(3^n).
\]

\end{proof}

% \subsection{\gha on Disjoint Unions of Graphs}

% \textbf{Disjoint Union of Graphs.}
% Let \( G_1 = (V_1, E_1) \) and \( G_2 = (V_2, E_2) \) be two graphs such that
% \( V_1 \cap V_2 = \emptyset \).
% The \emph{disjoint union} of \( G_1 \) and \( G_2 \), denoted
% \( G = G_1 \uplus G_2 \), is defined as
% \[
% G = (V_1 \cup V_2, E_1 \cup E_2).
% \]
% That is, \( G \) contains all vertices and edges of \( G_1 \) and \( G_2 \),
% and no edges between vertices of \( V_1 \) and \( V_2 \).

\begin{restatable}{theorem}{thmdisjoint}
\label{thm:disjoint}
Let \(G\) be a disjoint union of graphs \(G_1,G_2,\ldots,G_k\), where each \(G_i\) has
\(n_i\) agents and admits an algorithm for \gha running in time \(T\).
Then \gha on \(G\) can be solved in \(\mathcal{O}^\star(3^n\cdot T)\) time, where
\(n=\sum_{i=1}^k n_i\).
Moreover, if valuation functions are polynomially bounded integers, the running time
improves to \(\mathcal{O}^\star(2^n\cdot T)\).
\end{restatable}

% \begin{proof}

% Define a DP table \(D[i,S]\), where \(i\) denotes the first \(i\) components and \(S\subseteq H\) is the set of houses assigned to them.
% The entry \(D[i,S]\) stores the minimum total envy for allocating houses \(S\) to components \(G_1,\ldots,G_i\).

% The initialization is \(D[0,\emptyset]=0\).
% For \(i\ge 1\),

% $D[i,S]=\min_{R\subseteq S,\,|R|=n_i}
% \left(D[i-1,S\setminus R]+\mathcal{A}(G_i,R)\right)
% $,
% where \(\mathcal{A}(G_i,R)\) denotes the minimum envy achievable on \(G_i\) using the
% house set \(R\). We pre-compute all these entries $\mathcal{A}(G_i, R)$ in time at most $\mathcal{O}^\star(2^n \cdot T)$ by using the algorithm that runs in time $T$ on each connected component.

% A counting argument identical to Theorem~\ref{thm:tree} yields a total running time of
% \(\mathcal{O}^\star(3^n\cdot T)\), which can be improved to
% \(\mathcal{O}^\star(2^n\cdot T)\) using subset convolution when valuations are
% polynomially bounded.
% \end{proof}
\begin{proof}
We process the connected components one by one.

Define a dynamic programming table \( T \) as follows.
For \( i \ge 0 \) and \( S \subseteq H \), let \( T[i,S] \) denote the minimum
total envy achievable when allocating the houses in \( S \) to the first
\( i \) graphs \( G_1, \dots, G_i \).

Let \( \mathcal{A}(G_i, R) \) denote an algorithm that computes the minimum
total envy for the graph \( G_i \) when it is allocated the set of houses
\( R \subseteq H \).

\medskip
\noindent
\textbf{Base case.}
For \( i = 0 \), no graphs are processed. Hence,
\[
T[0, \emptyset] = 0,
\]
and for any nonempty \( S \subseteq H \), we set \( T[0,S] = \infty \).

\medskip
\noindent
\textbf{Recurrence (for \( i \ge 1 \)).}
For every subset \( S \subseteq H \),
\[
T[i, S]
=
\min_{\substack{R \subseteq S \\ |R| = |V(G_i)|}}
\Big(
T[i-1, S \setminus R] + \mathcal{A}(G_i, R)
\Big).
\]

Here, \( R \) represents the subset of houses assigned to graph \( G_i \),
whose size must equal the number of agents in \( G_i \).
The term \( T[i-1, S \setminus R] \) corresponds to the minimum total envy
for allocating the remaining houses to the first \( i-1 \) graphs.

\medskip
\noindent
\textbf{Proof of correctness.}

\medskip
\noindent
\textbf{Inductive hypothesis.}
Assume that for some \( i \ge 1 \), and for every subset \( S \subseteq H \),
the value \( T[i-1,S] \) correctly equals
\[
\mathrm{OPT}(i-1,S),
\]
the minimum total envy achievable when allocating the houses in \( S \)
to the graphs \( G_1,\dots,G_{i-1} \).

\medskip
\noindent
\textbf{Goal.}
We show that for every \( S \subseteq H \),
\[
T[i,S] = \mathrm{OPT}(i,S),
\]
where \( \mathrm{OPT}(i,S) \) denotes the minimum total envy when allocating
houses in \( S \) to graphs \( G_1,\dots,G_i \).

Fix a subset \( S \subseteq H \), and write
\( T := T[i,S] \) and \( \mathrm{OPT} := \mathrm{OPT}(i,S) \).
We prove both inequalities:
\[
T \ge \mathrm{OPT}
\quad \text{and} \quad
T \le \mathrm{OPT}.
\]

\medskip
\noindent
\paragraph{Part 1:  showing \( T \ge \mathrm{OPT} \).}

By the recurrence,
\[
T[i,S]
=
\min_{\substack{R \subseteq S \\ |R| = |V(G_i)|}}
\left(
T[i-1, S \setminus R] + \mathcal{A}(G_i, R)
\right).
\]

For every such subset \( R \), the value
\( \mathcal{A}(G_i,R) \) correctly computes the minimum envy for assigning
\( R \) to \( G_i \), and by the inductive hypothesis,
\[
T[i-1, S \setminus R] = \mathrm{OPT}(i-1, S \setminus R).
\]
Thus, each candidate in the minimization corresponds to a valid allocation
of \( S \) into two disjoint parts assigned to
\( G_1,\dots,G_{i-1} \) and \( G_i \), respectively.
Therefore, every candidate value is at least the true optimum, implying
\[
T[i,S] \ge \mathrm{OPT}(i,S).
\]

\medskip
\noindent
\paragraph{Part 2:  showing \( T \le \mathrm{OPT} \).}

Let \( R^* \subseteq S \) be the subset of houses assigned to \( G_i \)
in an optimal solution achieving \( \mathrm{OPT}(i,S) \). Then
\[
\mathrm{OPT}(i,S)
=
\mathrm{OPT}(i-1, S \setminus R^*)
+
\mathcal{A}(G_i, R^*).
\]

By the inductive hypothesis,
\[
T[i-1, S \setminus R^*]
=
\mathrm{OPT}(i-1, S \setminus R^*).
\]
Since the recurrence considers all subsets \( R \subseteq S \) of size
\( |V(G_i)| \), the subset \( R^* \) is included. Hence,
\[
T[i,S]
\le
T[i-1, S \setminus R^*] + \mathcal{A}(G_i, R^*)
=
\mathrm{OPT}(i,S),
\]
which contradicts the assumption that \( T[i,S] > \mathrm{OPT}(i,S) \).
Therefore,
\[
T[i,S] \le \mathrm{OPT}(i,S).
\]

\medskip
\noindent
\textbf{Conclusion.}
Since both inequalities hold, we conclude that
\[
T[i,S] = \mathrm{OPT}(i,S)
\quad \text{for all } i \ge 0 \text{ and } S \subseteq H.
\]
Thus, the dynamic programming recurrence correctly computes the minimum
total envy.

\paragraph{Time Complexity.}
The dynamic programming table contains \( i \cdot 2^n \) entries, where
\( i \le n \) and \( S \subseteq H \). Hence, the number of states is
\( \mathcal{O}^*(2^n) \).
For each entry, the recurrence enumerates all subsets
\( R \subseteq S \), leading to \( \mathcal{O}^*(2^n) \) work per state.
Therefore, the total running time is \( \mathcal{O}^*(3^n) \).

\paragraph{Subset convolution.}
To speed up the recurrence
\[
T[i, S]
=
\min_{\substack{R \subseteq S \\ |R| = |V(G_i)|}}
\left(
T[i-1, S \setminus R] + \mathcal{A}(G_i, R)
\right),
\]
we interpret it as a \emph{min-sum subset convolution}.
Define functions \( f,g : 2^H \to \mathbb{Z} \cup \{\infty\} \) by
\[
f(R) = T[i-1, R],
\qquad
g(R) = \mathcal{A}(G_i, R).
\]
Then the function
\[
h(S) = T[i,S]
\]
is exactly the min-sum subset convolution of \( f \) and \( g \).
Using standard algorithms for min-sum subset convolution, all values
\( T[i,S] \) can be computed in time \( \mathcal{O}^*(2^n) \).
This completes the proof of the theorem.

\end{proof}

It is known from \cite{DBLP:conf/atal/HosseiniPSVV23} and
\cite{DBLP:journals/aamas/BeynierCGHLMW19} that ME-GHA can be solved in polynomial
time when the social graph is a clique or a star, respectively. Therefore, we obtain the following corollary.

\begin{restatable}{corollary}{starclique}
\gha can be solved on disjoint unions of stars and cliques in
\(\mathcal{O}^\star(3^n)\) time. Moreover, if valuation functions are polynomially bounded integers, the running time improves to \(\mathcal{O}^\star(2^n)\)
\end{restatable}

\section{Outlook}

In this work, we initiate a systematic study of the algorithmic complexity of \ghafull, when the agents may have differing valuations for houses. To this end, we employed techniques from parameterized complexity and identified several natural parameterizations that makes the problem tractable, when the number of house types is bounded. In the most general scenario, without any restriction on the number of house types, we improve upon the trivial $\mathcal{O}^\star(n!)$ running time obtained via brute-force, we also designed \emph{moderately exponential-time} algorithms for \ghafull that run in $\mathcal{O}^\star(2^{O(n)})$ time in graphs of treewidth $O(n/\log n)$, specifically, for planar graphs. 

% Our work opens the door to several new directions and leaves some interesting questions. Perhaps the most interesting concrete question is whether  \gha on path graphs or disjoint edges with binary (and more generally, bi-valued) valuations is solvable in polynomial time. Another interesting case is designing tractable algorithms for (non-complete) bipartite graphs, where both sides are ``large''. At a broader level, it would be interesting to explore whether one can design \fpt-approximation algorithms for \gha that simultaneously (a) give improved approximation ratios compared to the polynomial-time regime~\cite{DBLP:conf/atal/Hosseini0SVV24}, and (b) are faster than the exact parameterized algorithms designed in this work. We leave this direction for a future work. 

Our work opens the door to several new directions and leaves some interesting
open questions. Perhaps the most intriguing concrete question is whether
\gha\ on path graphs or disjoint edges with binary (and more generally,
bi-valued) valuations is solvable in polynomial time. Another promising
direction is designing tractable algorithms for (non-complete) bipartite
graphs where both sides are ``large''. At a broader level, it would be
interesting to explore \fpt-approximation algorithms for \gha\ that
simultaneously (a) achieve improved approximation ratios compared to the
polynomial-time regime~\cite{DBLP:conf/atal/Hosseini0SVV24} and (b) run faster
than the exact parameterized algorithms developed in this work. 
Finally, an important open direction is to establish meaningful
\emph{lower bounds} for \gha\ under various parameterizations to understand the optimality of our algorithms.


\bibliographystyle{plainnat}   % or any style you want

\bibliography{archiveX}

\end{document}